\newcounter{myctr}
\def\myitem{\refstepcounter{myctr}\bibfont\noindent\ifnum\themyctr>9\else\phantom{0}\fi\hangindent17pt\themyctr.\enskip}

\documentclass{ws-ijqi}

\begin{document}

\markboth{F. Benatti, R. Floreanini}
{Entangled Fermions}

\catchline{}{}{}{}{}

\title{ENTANGLEMENT AND ALGEBRAIC INDEPENDENCE IN FERMION SYSTEMS}

\author{FABIO BENATTI}

\address{Dipartimento di Fisica, Universit\`a di Trieste,
34151 Trieste, Italy\\
Istituto Nazionale di Fisica Nucleare, Sezione di Trieste,
34151 Trieste, Italy\\
benatti@ts.infn.it}

\author{ROBERTO FLOREANINI}

\address{Istituto Nazionale di Fisica Nucleare, Sezione di Trieste,
34151 Trieste, Italy\\
florean@ts.infn.it}

\maketitle

\begin{history}
\received{Day Month Year}
\revised{Day Month Year}
\end{history}

\begin{abstract}
In the case of systems composed of identical particles, a typical instance in quantum statistical mechanics, the standard approach to separability and entanglement ought to be reformulated and rephrased in terms of correlations between operators from subalgebras localized in spatially disjoint regions. While this algebraic approach is straightforward for bosons, in the case of fermions it is subtler since one has to distinguish between micro-causality, that is the anti-commutativity of the basic creation and annihilation operators, and algebraic independence that is the commutativity of local observables. We argue that a consistent algebraic formulation of separability and entanglement should be compatible with micro-causality rather than with algebraic independence.
\end{abstract}

\keywords{Entanglement; Identical particles; Fermionic Systems}

\section{Introduction}	

In trying to apply the standard definitions of separability and entanglement to systems of identical
particles, one immediately faces a problem: the indistinguishability of the system constituents
conflicts with Hilbert space tensor product structure on which these notions are based.
The point is that the particles are identical and therefore they can not be singly addressed,
nor can their individual properties measured:
only collective, global system operators are in fact admissible,
experimentally accessible observables\footnote{Entanglement
in many-body systems has been widely discussed in the recent literature, {\it e.g.} see
\cite{Schliemann}-\cite{Modi}; however, for the reasons just pointed out,
only a limited part of those results are really applicable to the case of identical particle systems.} \cite{Feynman,Sakurai}.

In this context, the usually adopted
definition of separability based on the particle aspect of first quantization being too restrictive,
one resorts to the second quantization language proper to quantum many-body theory where the primary objects are the
algebras of operators rather than states in a Hilbert space \cite{Strocchi}.
The new point of view towards separability and entanglement has been advocated before \cite{Zanardi}-\cite{Viola2},
but formalized only recently \cite{Benatti1}-\cite{Benatti5} with particular attention on bipartite entanglement,
aiming at specific applications to quantum metrology.

In the following, we shall consider the second quantized (algebraic) approach to the notions of separability and entanglement
in the case of systems composed of fermions whose elementary creation and annihilation operators anti-commute.
We shall show that the canonical anti-commutation relations in connection with the properties of locality and commutativity
of the system observables make the theory of fermion entanglement even richer than in the case of bosonic systems.
Indeed, while anti-commutativity of creation and annihilation operators of orthogonal fermionic states corresponds to the axiom
of \textbf{micro-causality} in axiomatic quantum field theory~\cite{Strocchi,Streater,Strocchi2}, locality has to do with observables localized within
regions that forbid the possibility of interference between their respective measurements: these observables must then commute, a property known in the literature as \textbf{algebraic independence}~\cite{Summers}.
While for quantum systems consisting of bosons locality is compatible  with micro-causality as the creation and annihilation operators of single particle states  obey the canonical commutation relations, in the case of fermions is not so and this fact clearly emerges
when one wants to extend to such systems the standard notions of entanglement and separability.
In the following we define entanglement and separability in terms of micro-causality rather than basing on algebraic independence and argue that this a more consistent choice from a physical point of view.

\section{Entanglement in Fermi systems}

We shall consider a many-body system consisting of a fixed number $N$ of fermions each of which can be found in $M$ different modes, $i=1,2,\ldots,M\geq N$: the choice of modes is highly non-unique as they correspond  to the the orthogonal components of any orthonormal basis in the single particle ($M$-dimensional) Hilbert space, $M$ being possibly infinite. The second quantization description of such a system associates to each mode $i$ creation and annihilation
operators, $a_i^\dagger$, $a_i$~\cite{Strocchi} obeying the Canonical Anti-commutation relations (CAR)
\begin{equation}
\{a_i,\,a^\dagger_j\}\equiv a_i\,a^\dagger_j+a^\dagger_j\, a_i=\delta_{ij}\ ,\quad\{a_i,\,a_j\}=
\{a_i^\dagger,\,a^\dagger_j\}=\, 0\ .
\label{1}
\end{equation}
The most natural Hilbert space $\cal H$ associated to this system is the Fock space spanned by the states obtained by applying creation operators to the
vacuum vector $\vert 0\rangle$ ($a_i\vert 0\rangle=0$):
\begin{equation}
|n_1, n_2,\ldots,n_M\rangle=
(a_1^\dagger)^{n_1}\, (a_2^\dagger)^{n_2}\, \cdots\, (a_M^\dagger)^{n_M}\,|0\rangle\ ,
\label{2}
\end{equation}
the integers $n_1, n_2, \ldots, n_M$ representing the occupation numbers of the different modes;
due to (\ref{1}), they can take only the two values 0 or 1.
The set of polynomials in all creation and annihilation operators,
$\{a^\dagger_i,\, a_i\, |\,i=1,2,\ldots, M\}$,
form an algebra ${\cal A}$ of bounded operators acting on $\cal H$; the observables of the systems are part of this algebra.

In this setting the notions of separability and entanglement cannot just be extrapolated from the case of distinguishable particles.
In the case of two non-identical standard qubits these notions are connected  with the natural Hilbert space product structure ${\cal H}=\mathbb{C}^2\otimes\mathbb{C}^2$ and
the corresponding algebraic product structure for the space of the associated observables
${\cal A}=M_2(\mathbb{C})\otimes M_2(\mathbb{C})$, with $M_2(\mathbb{C})$ the set of $2\times2$ complex matrices.
Both are related to the individual particle picture whereby one is able to identify the two qubits and the local observables as those
taking the form
\begin{equation}
\label{3}
A\otimes B=(A\otimes 1)\,(1\otimes B)\ ,
\end{equation}
where $A$ is an observable of the first qubit, while $B$ that for the second one.
In other terms, local observables for  two-qubit systems are characterized by being tensor products of observables pertaining each to one of the two parties: they commute and, following \cite{Summers}, we term them as \textbf{algebraically independent}.

Consider instead a system composed by two fermions that can occupy two modes, and thus described
by the set of operators $(a_1, a_1^\dagger, a_2, a_2^\dagger)$:
the single particle Hilbert space is still $\mathbb{C}^2$; however, the CAR in (\ref{1}) make the total Hilbert space $\cal H$ consist of the vector $a^\dag_1a^\dag_2\vert 0\rangle$ only and the Fermi algebra $\cal A$  linearly generated by the identity
together with at most second order monomials in $a_1,a^\dag_1$ and $a_2,a_2^\dag$.
Clearly, the particle Hilbert space tensor product structure is lost as well as the locality of observables expressed
by the tensor product structure as in (\ref{3}).

The way out is provided  by identifying local observables with products of commuting observables that is with observables that
can be simultaneously and independently measured  without the need of attaching them to any particular particle

In quantum many-body theory, a most natural identification of local observables is in terms of self-adjoint operators supported within disjoint volumes, say a finite volume $V_1$ and its complement $V_2=\mathbb{R}^3\setminus V_1$ (disjoint apart from their common border). Then, one considers the two
subalgebras ${\cal A}_{1,2}$ generated by creation and annihilation operators $a(f_1)\,,\,a^\dag(f_1)$ and $a(f_2)\,,\, a^\dag(f_2)$ of normalized single particle states $f_1$, $f_2$ supported within the two volumes: $a^\dag(f_{1,2})\vert 0\rangle=\vert f_{1,2}\rangle$.

In the case of bosons, the Canonical Commutation Relations (CCR) yield $[a(f_1)\,,\,a(f_2)]=[a^\dag(f_1)\,,\,a^\dag(f_2)]=0$
and
\begin{equation}
\label{fermifields1}
[a(f_1)\,,\,a^\dag(f_2)]=\langle f_1\vert f_2\rangle=0\ ,\quad\forall\ f_{1,2}\ :\quad \hbox{supp}(f_1)\subseteq V_1 ,\ \hbox{supp}(f_2)\subseteq V_2\ .
\end{equation}
The vanishing commutators provide a non-relativistic expression of the so-called bosonic \textbf{micro-causality};  in  relativistic quantum field theory, micro-causality means that bosonic fields in causally-disjoint space-time regions cannot influence each other and must then commute~\cite{Streater,Strocchi2}.

Things radically change in the case of fermions; in this case, micro-causality demands  the anti-commutators $\{a(f_1)\,,\,a^\dag(f_2)\}$ to vanish.
On the other hand, the algebraic independence of operators is based on their vanishing commutators; in order to check that, one may use the algebraic relation
\begin{equation}
\label{6}
[AB\,,\,C]\,=\,A\,\{B\,,\,C\}-\{A\,,\,C\}\,B\ .
\end{equation}
It then follows that two operators supported in disjoint volumes commute when at least one of them is constructed by means polynomial involving only even powers of creation and annihilation operators.
Therefore, given two sub-algebras ${\cal A}_{1,2}$ of the Fermi algebra $\cal A$, localized within disjoint volumes, on one hand one has the micro-causality condition expressed by the anti-commutativity of the basic creation and annihilation operators,
\begin{equation}
\label{microcaus}
\{a^\#(f_1)\,,\,a^\#(f_2)\}=0 \qquad\forall\ f_{1,2}\ :\quad \hbox{supp}(f_1)\subseteq V_1 ,\ \hbox{supp}(f_2)\subseteq V_2\ ,
\end{equation}
where $a^\#$ stands for $a$ or $a^\dag$.
On the other hand, from the point of view of the algebraic independence of fermionic observables one ought to distinguish the so-called even and odd
components of ${\cal A}_{1,2}$.

\begin{definition}
Let $\Theta$ be the automorphism on the Fermi algebra ${\cal A}$
defined by $\Theta(a_i)=-a_i$, $\Theta(a_i^\dagger)=-a_i^\dagger$
for all $a_i,\ a^\dagger_i\in{\cal A}$.%
The even component $\mathcal{A}^e$ of $\mathcal{A}$ is the subset of elements $A^e\in\mathcal{A}$
such that $\Theta(A^e)=A^e$, while the odd component $\mathcal{A}^o$ of $\mathcal{A}$ consists
of those elements $A^o\in\mathcal{A}$ such that $\Theta(A^o)=-A^o$.
\end{definition}

\noindent
\textbf{Remark 1.}\quad
The even component $\mathcal{A}^e$ is generated by the norm closure of even polynomials
in creation and annihilation operators and is a  subalgebra of $\cal A$, while the odd component $\mathcal{A}^o$ is only a linear space
since the product of two odd elements is even.
Even if self-adjoint, odd elements like
\begin{equation}
\label{fermifields2}
A^o_1=a(f_1)+a^\dag(f_1)\ ,\quad A_2^o=a(f_2)+a^\dag(f_2)\ ,
\end{equation}
with $a^\#(f_{1,2})$ as in~\eqref{microcaus} are not considered to be observable as they are not compatible with superselection rules~\cite{Bartlett,Wick,Moriya}:
for instance, they do not leave the number operators invariant.
Since they do not commute, were they observable, their respective measurements would interfere with each other despite the disjointness of their supports.
In axiomatic relativistic quantum field theory, they are known as unobservable fields: however, with their even powers one constructs operators like energy and currents. These are not only observable, but, if supported within causally separated regions, they also commute and are thus algebraically independent.
\hfill$\Box$
\medskip

The splitting of the whole algebra $\cal A$ into a bipartition consisting of two subalgebras supported within disjoint volumes can be generalized by means of annihilation and creation operators corresponding to different modes.
Indeed, a bipartition of the algebra ${\cal A}$ of a system of $N$ fermions each one capable of $M\geq N$ modes can be given by splitting the collection
of creation and annihilation operators into two disjoint sets, $\{a_i^\dagger,\, a_i\, | i=1,2\ldots,m\}$ and
$\{a_j^\dagger,\, a_j,\, |\, j=m+1,m+2,\ldots,M\}$;
it is thus uniquely determined by the choice of the integer $m$, with $0\leq m \leq M$.%

In order to discuss the consequences of the second-quantization (algebraic) approach we start with the following definitions:

\begin{definition}
\begin{enumerate}
\item
Two subalgebras ${\cal A}_1$, ${\cal A}_2$ of the Fermi lagebra $\cal A$ will be called disjoint if they are generated by the norm-closure of polynomials in annihilation and creation operators of modes belonging to disjoint subsets $I_1$ and $I_2$.
\item
An {\bf algebraic bipartition} of the Fermi  algebra $\cal A$ is any pair
$({\cal A}_1, {\cal A}_2)\subset {\cal A}$ of disjoint subalgebras, with only the identity operator in common, such that
${\cal A}_1 \cup {\cal A}_2={\cal A}$.
\item
An operator of $\cal A$ is said to be
$({\cal A}_1, {\cal A}_2)$-{\bf local}, {\it i.e.} local with respect to
a given bipartition $({\cal A}_1, {\cal A}_2)$, if it is the product $A_1 A_2$ of an element
$A_1$ of ${\cal A}_1$ and another $A_2$ in ${\cal A}_2$.
\end{enumerate}
\end{definition}

In general, a state $\omega$ over the Fermi algebra $\cal A$ is any normalized, positive, linear (expectation) functional $\omega:{\cal A}\mapsto \mathbb{C}$, such that
the average value of any observable $\cal O$ can be expressed as the value taken by $\omega$ on it,
$\langle {\cal O}\rangle=\omega({\cal O})$, the standard example being
$\langle {\cal O}\rangle={\rm Tr}(\rho\,{\cal O})$,
namely an expectation functional given by the trace operation with respect to a density matrices $\rho$.

From the notion of operator locality, a natural definition of state separability (absence of non-local correlations)
and entanglement (presence of non-local correlations) follows \cite{Benatti1}:
\medskip

\noindent
\begin{definition}
\label{def3}
A state $\omega$ on the algebra ${\cal A}$ will be called {\bf separable} with
respect to the bipartition $({\cal A}_1, {\cal A}_2)$ if the expectation $\omega(A_1 A_2)$
of any local operator $A_1 A_2$ can be decomposed into a linear convex combination of
products of expectations:
\begin{equation}
\omega(A_1 A_2)=\sum_k\lambda_k\, \omega_k^{(1)}(A_1)\, \omega_k^{(2)}(A_2)\ ,\qquad
\lambda_k\geq0\ ,\quad \sum_k\lambda_k=1\ ,
\label{4}
\end{equation}
where $\omega_k^{(1)}$ and $\omega_k^{(2)}$ are given states on ${\cal A}$;
otherwise the state $\omega$ is said to be {\sl entangled} with respect the bipartition
$({\cal A}_1, {\cal A}_2)$.%
\end{definition}

\noindent
\textbf{Remark 2.}\quad
It clearly appears from the previous definition that separability or its absence are properties of states of systems of identical particles which strongly depend on the chosen bipartition. Indeed, as already remarked in the introduction, there is no a-priori given algebraic split into system $1$ and system $2$ as in the case of the tensor product of the algebras of two distinguishable particles~\cite{Zanardi}-\cite{Benatti5};
this general observation, often overlooked, is at the origin of
much confusion in the recent literature.\hfill$\Box$

\section{Separable and entangled fermionic states}

For bosonic states \cite{Benatti1}, the two  subalgebras ${\cal A}_1$, ${\cal A}_2$ commute. As already observed, the condition $[A_1,\, A_2]=\,0$ for all $A_i\in{\cal A}_i$, $i=1,2$, encodes at the algebraic level the intuition that entanglement should be connected with the presence of non-classical correlations among
{\sl commuting}, that is algebraically independent, observables.%
\footnote{For this reasons, in dealing with fermion systems, the discussion is often restricted
just to the commuting subalgebras $\mathcal{A}_1^e$, $\mathcal{A}_2^e$
of even operators \cite{Balachandran}}.
However, in the case of fermions, the operators $A_{1,2}$ are only required to satisfy the condition of fermionic micro-causality, namely that they must belong to subalgebras constructed by anti-commuting annihilation and creation operators. In the following we shall clarify the reasons for this choice.

Given the algebraic bipartition $(\mathcal{A}_1,\, \mathcal{A}_2)$,
one can define the even $\mathcal{A}_i^e$ and odd $\mathcal{A}_i^o$ components of the two
subalgebras $\mathcal{A}_i$, $i=1,2$.
Only the operators of the first partition
belonging to the even component $\mathcal{A}_1^e$
commute with any operator of the second partition and, similarly, only the even operators
of the second partition commute with the whole subalgebra $\mathcal{A}_1$.

\subsection{Structure of separable fermionic states}

A crucial   observation is that the decomposition in~\eqref{4} makes sense only when
at least one of the state entering each of the products at the right hand side vanishes on odd elements.
This fact follows from a result~\cite{Araki} whose simple proof we report as it sheds light
upon the constraints posed by anti-commutativity.
\begin{lemma}
Consider a bipartition $(\mathcal{A}_{1},\mathcal{A}_{1})$ of the fermion algebra $\mathcal{A}$
and two states $\omega_1$, $\omega_2$ on $\mathcal{A}$. Then, the linear
functional $\omega$ on $\mathcal{A}$ defined by
$\omega(A_1A_2)=\omega_1(A_1)\,\omega_2(A_2)$ for all $A_1\in\mathcal{A}_1$ and $A_2\in\mathcal{A}_2$ is a
state on $\mathcal{A}$ only if at least one $\omega_i$ vanishes on the odd component of $\mathcal{A}_i$.
\end{lemma}
\begin{proof}
If $\omega_{1,2}$ do not vanish on the odd components $\mathcal{A}^o_{1,2}$, there exist  self-adjoint $A^o_i\in\mathcal{A}_i^o$, such that $\omega_i(A^o_i)\neq 0$, $i=1,2$.
Then, the anti-commutativity of the odd elements $A^o_i$ yields a contradiction as
$$
\overline{\omega(A^o_1A^o_2)}= \omega(A^o_2A^o_1)=-\omega(A^o_1A^o_2)=\omega_1(A^o_1)\,\omega(A^o_2)\neq 0\ .
$$
\end{proof}

It thus turns out that, given a bipartition $({\cal A}_1,{\cal A}_2)$ of the fermion algebra $\cal A$,
{\it i.e.} a decomposition of $\cal A$ in the subalgebra ${\cal A}_1$ generated by the first
$m$ modes and the subalgebra ${\cal A}_2$, generated by the
remaining $M-m$ ones, the decomposition (\ref{4}) is meaningful only for local operators
$A_1 A_2$ for which $[A_1,\, A_2]=\,0$, so that, also for fermions, separable states yield linear convex combination of products of mean values on all products of commuting observables.

\subsection{Structure of entangled fermionic states}

As a consequence of the previous Lemma, we also have that if a state $\omega$ on the Fermi algebra $\cal A$ does not vanish on a local operator $A_1^o A_2^o$, with the two components
$A_1^o \in {\cal A}_1^o$, $A_2^o \in {\cal A}_2^o$ both belonging to the odd part
of the two subalgebras, then it is entangled. Indeed, in such a case it cannot be split as in Definition \ref{def3}.

Given a bipartition $({\cal A}_1,{\cal A}_2)$ where the number of modes is $M=2N$, with odd number of fermions $N$ and ${\cal A}_1$, respectively ${\cal A}_2$, is constructed with creation and annihilation operators of the first $N$, respectively the second $N$ modes, a very simple instance of a state with the above characteristics is given by a pure state consisting of the balanced superposition of $N$ fermions in the first $N$ modes and none in the other ones, with no fermions in the first $N$ modes and $N$ in the second ones:
\begin{equation}
\label{entst1}
\vert\Psi\rangle=\frac{1}{\sqrt{2}}\Big(\vert N;0\rangle+\vert0;N\rangle\Big)=\frac{1}{\sqrt{2}}\Big(
a^\dag_1a^\dag_2\cdots a_N^\dag+a^\dag_{N+1}a^\dag_{N+2}\cdots a_{2N}^\dag\Big)\vert 0\rangle\ .
\end{equation}
Consider the product of odd elements $A^0_1A^0_2$, where $A_1^o=a_1a_2\cdots a_N\in{\cal A}_1$ and $A_2^o=a^\dag_{N+1}a^\dag_{N+2}\cdots a^\dag_{2N}\in{\cal A}$; then
\begin{equation}
\label{entst2}
A^0_1A^0_2\vert 0;N\rangle=0\ ,\quad A^0_1A^0_2\vert N;0\rangle\propto\vert0;N\rangle\ \Longrightarrow\ \langle\Psi\vert A_1^oA_2^o\vert \Psi\rangle=\frac{1}{2}\ .
\end{equation}
Thus, in full agreement with its evident entangled structure, the pure state $\vert\Psi\rangle$ is not separable according to the Definition \ref{def3}.

However, if state separability had been defined by asking the factorization in~\eqref{4} only relatively to even (and therefore commuting) operators, it would have followed that, on the algebra generated by ${\cal A}^e_{1,2}$,  the pure state $\vert\Psi\rangle$ coincides with the separable density matrix
$$
\rho_{sep}=\frac{1}{2}\vert N;0\rangle\langle N;0\vert+\frac{1}{2}\vert 0;N\rangle\langle 0;N\vert\ .
$$
Indeed, $N$ odd implies $\langle N;0\vert A^e_1\vert 0;N\rangle=\langle N;0\vert A^e_2\vert 0;N\rangle=0$ for all $A_{1,2}^e$, whence
\begin{eqnarray}
\nonumber
\langle\Psi\vert A_1^eA_2^e\vert \Psi\rangle&=&\frac{1}{2}\Big(\langle N;0\vert A_1^e\vert N;0\rangle\, \langle N;0\vert A_2^e\vert N;0\rangle\, +\,
\langle 0;N\vert A_1^e\vert 0;N\rangle\, \langle 0;N\vert A_2^e\vert 0;N\rangle \Big)\\
&=&{\rm Tr}\Big(\rho_{sep}A^e_1A_2^e\Big)\ ,\qquad \forall A_1^e\in{\cal A}_1^e\ ,\  A_2^e\in{\cal A}_2^e\ .
\label{entst3}
\end{eqnarray}

In relativistic quantum field theory a still open problem is the  relation between the locality of observables  (identified with them being algebraically independent) and their statistical independence which is related to a reference state \cite{Summers,buchsumm}.
More concretely, the issue at stake is to derive the commutativity of two sub-algebras ${\cal A}_{1,2}$ from whether or not the sufficiently many states factorize over product of observables $A_1\in{\cal A}_1$ and $A_2\in{\cal A}_2$.
In this context one has the following definition of uncorrelated states~\cite{buchsumm}.
\medskip

\begin{definition}
Given two (not necessarily disjoint) subalgebras ${\cal A}_{1,2}$ of the Fermi algebra $\cal A$,  a state $\omega$ on $\mathcal A$ is $({\cal A}_1,{\cal A}_2)$-uncorrelated if
\begin{equation}
\label{def4}
\omega(P_1\wedge P_2) = \omega(P_1) \omega(P_2)\ ,
\end{equation}
for every pair of projections $P_1\in{\cal A}_1$ and $P_2\in{\cal A}_2$,
where
\begin{eqnarray}
\nonumber
P_1\wedge P_2&=&\lim_{n\to+\infty}(P_1P_2)^n=\lim_{n\to+\infty} (P_2P_1)^n\\
\label{largproj}
&=&\lim_{n\to+\infty}P_1(P_1P_2P_1)^nP_1=\lim_{n\to+\infty}P_2(P_2P_1P_2)^nP_2
\end{eqnarray}
denotes the largest projector $Q\in{\cal A}$ such that $Q\leq P_1$, $Q\leq P_2$.
\end{definition}

The above definition refers to any pair of subalgebra, that is not necessarily forming an algebraic bipartition in the sense of Definition~\ref{def3}.
It is thus interesting to relate the entangled structure of $\vert\Psi\rangle$ to the above characterization of uncorrelated states: it turns out that, given a bipartition $({\cal A}_1,{\cal A}_2)$, the pure $({\cal A}_1,{\cal A}_2)$-entangled state $\vert\Psi\rangle$ in~\eqref{entst1} is also $({\cal A}_1,{\cal A}_2)$-correlated. This enforces the need of defining fermionic entanglement basing on micro-causality rather than on algebraic independence.

This can be seen by considering the projections
$$
P_1=\frac{1+a_1+a^\dag_1}{2}\in{\cal A}_1 ,\quad P_2=\frac{1+a_{N+1}+a^\dag_{N+1}}{2}\in{\cal A}_2\ ,
$$
constructed with the creation and annihilation operators of the first and $N+1$-th mode that belong to the different subalgebras of the disjoint
pair $({\cal A}_1,{\cal A}_2)$.
It proves convenient to work within the spin-representation provided by the Jordan-Wigner transformation
\begin{eqnarray}
\label{JW1}
a_i&=&\Big(\bigotimes_{j=1}^{i-1}\sigma_z\Big)\otimes\,\sigma_-\,\otimes\Big(\bigotimes_{i+1}^{2N} 1\Big)\ ,\quad
\sigma_-=\frac{\sigma_x-i\sigma_y}{2}\quad\hbox{whereby}\\
\label{JW2}
P_1&=&\frac{1+\sigma_x}{2}\otimes\Big(\bigotimes_{j=2}^{2N}1\Big)\\
\label{JW3}
P_2&=&\frac{1}{2}\Bigg(\Big(\bigotimes_{j=1}^{2N}1\Big)\,+\,
\Big(\bigotimes_{j=1}^{N}\sigma_z\Big)\otimes\sigma_x\otimes\Big(\bigotimes_{j=N+2}^{2N}1\Big)\Bigg)\ .
\end{eqnarray}
Then, one computes
\begin{eqnarray}
\nonumber
P_1P_2P_1&=&\frac{1+\sigma_x}{4}\otimes\Big(\bigotimes_{j=2}^{2N}1\Big)\ .
\end{eqnarray}
The non-zero eigenvalue of $P_1P_2P_1$ is $1/2$: thus $P_1\wedge P_2=0$, $\langle\Psi\vert P_1\wedge P_2\vert\Psi\rangle=0$, while
\begin{equation}
\label{JW5}
\langle\Psi\vert P_1\vert\Psi\rangle=\langle\Psi\vert P_2\vert\Psi\rangle=\frac{1}{4}\ ,
\end{equation}
so that $\vert\Psi\rangle$ cannot fulfill the condition~\eqref{def4} and is thus $({\cal A}_1,{\cal A}_2)$-correlated.

\eject

\section{Conclusions}

In order to properly extend the notions of separability and entanglement to quantum systems consisting of identical particles, an optimal framework is provided by the second quantization formalism. This allows one to resort to the more general mode picture and to abandon the particle based tensor product structure of Hilbert spaces and algebras of observables, valid only for distinguishable particle.

Unlike in the case of bosonic systems, in the case of fermions the fundamental anti-commutativity of creation and annihilation operators of orthogonal modes,
known as micro-causality, conflicts with the notion of algebraic independence of local observables, that is with the fact that they must commute in order to
be simultaneously measurable without interferences.
In order to reconcile micro-causality with algebraic independence one usually restricts oneself to considering even fermionic subalgebras, namely the closures of even polynomials in fermionic creation and annihilation operators. Then, the elements of these subalgebras, even if constructed with anti-commuting creation and annihilation operators of orthogonal modes, do commute

We have shown that, if separability is defined by restricting to commuting even subalgebras, that is to algebraic independence, then an apparently entangled state would be termed separable, whereas the same state is perfectly entangled with respect to a definition of absence of correlations based on micro-causality, namely with reference to disjoint full fermionic subalgebras, that is not only to their even components.

\end{document}